\documentclass[11pt]{article}

\usepackage[usenames,dvipsnames]{xcolor}
\definecolor{Gred}{RGB}{219, 50, 54}
\definecolor{ToCgreen}{RGB}{0, 128, 0}

\usepackage[scale=0.97]{XCharter} 

\usepackage[libertine,bigdelims,vvarbb,scaled=1.05]{newtxmath} 

\usepackage{babel}
\usepackage[spacing=true,kerning=true,babel=true,tracking=true]{microtype}

\DeclareMathAlphabet{\pazocal}{OMS}{zplm}{m}{n} 
\renewcommand{\mathcal}[1]{\pazocal{#1}}

\usepackage{textcomp} 
\usepackage{amsmath}
\usepackage{amsfonts}
\usepackage{pifont}
\usepackage{algorithm}
\usepackage{algorithmicx}
\usepackage[noend]{algpseudocode}
\usepackage{amsthm}
\usepackage{macros}
\usepackage{fullpage}
\usepackage{thmtools}
\usepackage{dsfont}
\usepackage{braket}
\usepackage{todonotes}
\usepackage[symbol]{footmisc}

\usepackage{type1cm}
\usepackage{pifont}
\makeatletter
\def\@fnsymbol#1{\ensuremath{\ifcase#1\or *\or \dagger\or \ddagger\or
   \mathsection\or \mathparagraph\or \|\or **\or \dagger\dagger
   \or \ddagger\ddagger \else\@ctrerr\fi}}
\makeatother

\DeclareMathOperator{\calD}{\mathcal{D}}
\DeclareMathOperator{\calE}{\mathcal{E}}
\DeclareMathOperator{\Tr}{\mathrm{Tr}}
\newcommand{\norm}[1]{\|#1\|}

\hypersetup{
  colorlinks=true,
  citecolor=ToCgreen,
  linkcolor=Sepia,
  filecolor=Gred,
  urlcolor=Gred
  }

\title{Predicting quantum channels over general product distributions}

\author{
Sitan Chen\thanks{Harvard University. Email: \href{mailto:sitan@seas.harvard.edu}{sitan@seas.harvard.edu}.}
\and
Jaume de Dios Pont\thanks{ETH Zurich. Email: \href{mailto:jaume.dediospont@math.ethz.ch}{jaume.dediospont@math.ethz.ch}.}
\and
Jun-Ting Hsieh\thanks{Carnegie Mellon University. Email: \href{mailto:juntingh@cs.cmu.edu}{juntingh@cs.cmu.edu}.}
\and
Hsin-Yuan Huang\thanks{Google Quantum AI, Caltech. Email: \href{mailto:hsinyuan@google.com}{hsinyuan@google.com}, \href{mailto:hsinyuan@caltech.edu}{hsinyuan@caltech.edu}.}
\and
Jane Lange\thanks{MIT. Email: \href{mailto:jlange@mit.edu}{jlange@mit.edu}.}
\and
Jerry Li\thanks{Microsoft Research. Email: \href{mailto:jerryzli@u.washington.edu}{jerryzli@u.washington.edu}.}
}

\date{}

\begin{document}
\maketitle

\begin{abstract}
We investigate the problem of predicting the output behavior of unknown quantum channels. Given query access to an $n$-qubit channel $\mathcal{E}$ and an observable $\mathcal{O}$, we aim to learn the mapping
\begin{equation*}
    \rho \mapsto \Tr(\mathcal{O} \mathcal{E}[\rho])
\end{equation*}
to within a small error for most $\rho$ sampled from a distribution $\mathcal{D}$. Previously, Huang, Chen, and Preskill~\cite{huang2023learning} proved a surprising result that even if $\mathcal{E}$ is arbitrary, this task can be solved in time roughly $n^{O(\log(1/\epsilon))}$, where $\epsilon$ is the target prediction error. 
However, their guarantee applied only to input distributions $\mathcal{D}$ invariant under all single-qubit Clifford gates, and their algorithm fails for important cases such as general product distributions over product states $\rho$. 

In this work, we propose a new approach that achieves accurate prediction over essentially any product distribution $\mathcal{D}$, provided it is not ``classical'' in which case there is a trivial exponential lower bound. Our method employs a ``biased Pauli analysis,'' analogous to classical biased Fourier analysis. Implementing this approach requires overcoming several challenges unique to the quantum setting, including the lack of a basis with appropriate orthogonality properties. The techniques we develop to address these issues may have broader applications in quantum information.

\end{abstract}

\section{Introduction}

When is it possible to learn to predict the outputs of a quantum channel $\calE$?
Such questions arise naturally in a variety of settings, such as the experimental study of complex quantum dynamics~\cite{huang2022quantum,huang2021information}, and in fast-forwarding simulations of Hamiltonian evolutions~\cite{cirstoiu2020variational,gibbs2024dynamical}.
However, in the worst case this problem is intractable, as it generalizes the classical problem of learning an arbitrary Boolean function over the uniform distribution from black-box access.
To circumvent this, our goal is to understand families of natural restrictions on the problem under which efficient estimation is possible.

One way to avoid this exponential scaling would be to posit further structure on the channel, e.g. by assuming it is given by a shallow quantum circuit~\cite{nadimpalli2023pauli,huang2024learning} or a structured Pauli channel \cite{harper2020efficient, flammia2021pauli, harper2021fast, van2023probabilistic, arunachalam2024learning}.
However, there are settings where the evolutions may be quite complicated --- e.g. the channel might correspond to the time evolution of an evaporating black hole~\cite{hayden2007black, hayden2019learning, penington2022replica, yang2023complexity} --- and where it is advantageous to avoid such strong structural assumptions on the underlying channel.



Recently,~\cite{huang2023learning} considered an alternative workaround in which one only attempts to learn a complicated $n$-qubit channel in an \emph{average-case} sense. Given query access to $\calE$, and given an observable $\mcO$, the goal is to learn the mapping
\begin{equation*}
    \rho \mapsto \Tr(\mcO\calE[\rho]) 
\end{equation*}
accurately on average over input states $\rho$ drawn from some $n$-qubit distribution $\calD$, rather than over worst-case input states. The authors of~\cite{huang2023learning} came to the surprising conclusion that this average-case task is tractable even for \emph{arbitrary} channels $\calE$, provided $\calD$ comes from a certain class of ``locally flat'' distributions. Their key observation was that the Heisenberg-evolved observable $\calE^{\dagger}[\mcO]$ admits a \emph{low-degree approximation} in the Pauli basis, where the quality of approximation is defined in an average-case sense over input state $\rho$.

Another interesting feature of this result is that their learning algorithm only needs to query $\calE$ on random product states, regardless of the choice of locally flat distribution $\calD$. This is both an advantage and a shortcoming. On one hand, if one is certain that the states $\rho$ one wants to predict on are samples from a locally flat distribution, no further information about $\calD$ is needed to implement the learning protocol in~\cite{huang2023learning}. On the other hand, locally flat distributions are quite specialized: they are constrained to be invariant under any single-qubit Clifford gate. In particular, almost all product distributions over product states fall outside this class. Worse yet, the general approach of low-degree approximation in the Pauli basis can be shown to fail when local flatness does not hold (see \Cref{sec:lowerbound_unbiased_degree}). We therefore ask:

\begin{center}
    {\em 
        Are there more general families of distributions $\calD$ under which one can\\
        learn to predict arbitrary quantum dynamics?
    }
\end{center}

Identifying rich settings where it is possible to characterize the average-case behavior of such dynamics, while making minimal assumptions on the dynamics, is of intense practical interest. Unfortunately, our understanding of this remains limited: even for general product distributions, known techniques break down. In this work we take an important first step towards this goal by completely characterizing the complexity of learning to predict arbitrary quantum dynamics in the product setting.
Informally stated, our main result is that learning is possible \emph{so long as the distribution is not classical}.
That is, for this problem there is a ``blessing of quantum-ness'': as long as the distribution displays any quantitative level of quantum behavior, there is an efficient algorithm for predicting arbitrary quantum dynamics under this distribution.

More formally, note that if $D$ is the uniform distribution over the computational basis states $\ket{0}$ and $\ket{1}$, then the task of predicting $\Tr(\mcO\calE[\rho])$ on average over $\rho\sim \calD \triangleq D^{\otimes n}$ for an arbitrary channel $\calE$ is equivalent to the task of learning an arbitrary Boolean function from random labeled examples, which trivially requires exponentially many samples. This logic naturally extends to any ``two-point'' distribution in which $D$ is supported on two diametrically opposite points on the Bloch sphere. 
Note that any such distribution, up to a rotation, is an embedding of a classical distribution onto the Bloch sphere.

A natural way of quantifying closeness to such distributions is in terms of the second moment matrix $\mcS \in \R^{3\times 3}$ of the distribution $D$, when $D$ is viewed as a distribution over the Bloch sphere (see \Cref{sec:bloch} for formal definitions).
We refer to this matrix as the \emph{Pauli second moment matrix} of $D$.
For the purposes of this discussion, the key property of this matrix is that $\norm{\mcS}_{\mathsf{op}} \leq 1$ for all $D$, and moreover, $\norm{\mcS}_{\mathsf{op}} = 1$ if and only if $D$ is one of the aforementioned two-point distributions.
With this, we can now state our main result:

\begin{theorem}[Learning an unknown quantum channel]
\label{thm:quantum}
    Let $\eps, \delta, \eta \in (0,1)$.
   Let $D$ be an unknown distribution over the Bloch sphere with Pauli second moment matrix $\mcS$ such that $\norm{\mcS}_{\op} \le 1-\eta$.
   Let $\mcE$ be an unknown $n$-qubit quantum channel, and let $\mcO$ be a known $n$-qubit observable.
   There exists an algorithm with time and sample complexity
   $\min(2^{O(n)}/\epsilon^2, n^{O(\log(1/\eps)/\log(1/(1-\eta)))})\cdot \log(1/\delta)$ that outputs an efficiently computable map $f'$ such that 
   \[\E_{\rho \sim D^{\otimes n}} [(\tr(\mcO\mcE[\rho]) - \tr(f'(\rho)))^2] \le \eps\]
   with probability at least $1-\delta$. 
\end{theorem}
\noindent Note that the only condition on $D$ we require is a quantitative bound on the spectral norm of its Pauli second moment matrix.
In other words, so long as the distribution $D$ is far from any two-point distribution, i.e., it is far from any classical distribution, we demonstrate that there is an efficient algorithm for learning to predict general quantum dynamics under this distribution.
Previously it was only known how to achieve the above guarantee in the special case where $D$ has mean zero. Indeed, as soon as one deviates from the mean zero case, the Pauli decomposition approach of~\cite{huang2023learning} breaks down. In contrast, our guarantee works for any product distribution whose marginal second moment matrices have operator norm bounded away from $1$.

\begin{remark}
We note that our techniques generalize to the case where the distribution is the product of different distributions over qubits, 
so long as each distribution has second moment with operator norm bounded by $1-\eta$.
However, for readability we will primarily focus on the case where all of the distributions are the same.
See \Cref{rem:different-distributions} for a discussion of how to easily generalize our techniques to this setting.
\end{remark}

\paragraph{Beyond low-degree concentration in an orthonormal basis.} Here we briefly highlight the key conceptual novelties of our analysis, which may be of independent interest. We begin by recalling the analysis in~\cite{huang2023learning} in greater detail. They considered the decomposition of $\mcO'\triangleq \calE^{\dagger}[\mcO]$ into the basis of $n$-qubit Pauli operators, i.e. $\mcO' = \sum_{P \in \{I,X,Y,Z\}^n} \alpha_P\cdot P$ and argued that this is well-approximated by the \emph{low-degree truncation} $\mcO'_{\rm low} = \sum_{|P|<t} \alpha_P \cdot P$. This can be readily seen from the following calculation. By rotating the distribution $D$, we may assume the covariance is diagonal, with entries bounded by $1 - \eta$. Then the error achieved by the low-degree truncation is given by
\begin{equation*}
    \E[\Tr((\mcO' - \mcO'_{\rm low})\rho)^2] = \E\Bigl[\Bigl(\sum_{|P|\ge t} \alpha_P\cdot \Tr(P\rho)\Bigr)^2\Bigr] \le \sum_{|P|\ge t} (1 - \eta)^{|P|} \cdot \alpha_P^2 \le (1 - \eta)^t \cdot \frac{1}{2^n}\|\mcO'\|^2_F \le (1 - \eta)^t \,,
\end{equation*}
where the second step follows from the fact that
$\E[\Tr(P\rho) \Tr(Q\rho)] = 0$ if $P \neq Q$ and is at most $(1-\eta)^{|P|}$ if $P = Q$ (since $D$ is mean zero and its covariance is diagonal with entries bounded by $1 - \eta$), and the last step follows by the assumption that $\|\mcO'\|_{\sf op} \le 1$.

Note that when $D$ is not mean zero, this step breaks, and we do not have this nice exponential decay in $t$. In fact, in \Cref{sec:lowerbound_unbiased_degree} we construct examples of operators which are not well-approximated by their low-degree truncations in the Pauli basis when $D$ has mean bounded away from zero.

A natural attempt at a workaround would be to change the basis under which we truncate. At least classically, biased product distributions over the Boolean hypercube still admit suitable orthonormal bases of functions, namely the \emph{biased Fourier characters}. As we show in \Cref{sec:classical}, this idea can be used to give the following learning guarantee in the classical case where $D$ is only supported along the $Z$ direction in the Bloch sphere:

\begin{theorem}[PAC learning over a concentrated product distribution]
\label{thm:classical}
    Let $\eps, \delta, \eta \in (0,1)$.
    Let $D$ be an unknown distribution over the interval $[-(1-\eta),1-\eta]$.
    Let $f:[-1,1]^n \to [-1,1]$ be an unknown bounded, multilinear function.
    There exists an algorithm with time and sample complexity 
   $n^{O(\log(1/\eps)/\log(1/(1-\eta)))}\cdot \log(1/\delta)$ that outputs a hypothesis $f'$ such that 
   \[\E_{x \sim D^{\otimes n}}[(f(x) - f'(x))^2] \le \eps\]
   with probability at least $1-\delta$. 
\end{theorem}

\noindent Unfortunately, when we move beyond the classical setting, the picture becomes trickier. In particular, it is not immediately clear what the suitable analogue of the biased Fourier basis should be in the quantum setting. We could certainly try to consider single-qubit operators of the form $\wt{P} = P - \mu_P\cdot I$ for $P\in\{X,Y,Z\}$, where $\mu_P$ denotes the $P$-th coordinate of the mean of $D$ regarded as a distribution over the Bloch sphere. We could then extend naturally to give a basis over $n$ qubits, and the functions $\rho\mapsto \Tr(\wt{P}\rho)$ would by design be orthogonal to each other with respect to the distribution $D^{\otimes n}$. Writing $\mcO' = \sum_P \wt{\alpha}_P \cdot \wt{P}$ and defining $\wt{\mcO}_{\rm low} \triangleq \sum_{|P|<t}\wt{\alpha}_P \cdot \wt{P}$, we can mimic the calculation above and obtain
\begin{equation*}
    \E[\Tr((\mcO' - \wt{\mcO}_{\rm low})\rho)^2] = \E\Bigl[\Bigl(\sum_{|P|\ge t} \wt{\alpha}_P\cdot \Tr(\wt{P}\rho)\Bigr)^2\Bigr] \le (1 - \eta)^{|P|} \sum_{|P|\ge t} \wt{\alpha}^2_P\,.
\end{equation*}
Unfortunately, at this juncture the above naive approach hits a snag. In the mean zero case, we could easily relate $\sum_P \alpha_P^2$ to $\frac{1}{2^n}\norm{\mcO'}^2_F$ because of a fortuitous peculiarity of the mean-zero setting. In that setting, we implicitly exploited both that the Pauli operators $P$ are orthogonal to each other with respect to the trace inner product, and also that that the functions $\rho\mapsto \Tr(P\rho)$ are orthogonal to each other with respect to $D^{\otimes n}$. In the above approach for nonzero mean, we achieved the latter condition by shifting the Pauli operators to define $\wt{P}$, but these shifted operators are no longer orthogonal to each other with respect to the trace inner product.

Circumventing this issue is the technical heart of our proof. As we will see in \Cref{sec:quantum}, several key technical moves are needed. First, instead of assuming that the covariance of $D$ is diagonalized, we will fix a rotation that simplifies the \emph{mean}, $\E[\rho]$. Second, instead of shifting the basis operators $\{X,Y,Z\}$ so that the resulting functions $\rho\mapsto \Tr(\wt{P}\rho)$ are orthogonal with respect to $D^{\otimes n}$, we shift them so that $\E[\rho]$ is \emph{orthogonal} to them, and define $\mcO'_{\rm low}$ by truncating in this new basis instead. Finally, instead of directly bounding the truncation error $\E[\Tr((\mcO' - \mcO'_{\rm low})\rho)^2]$ using the above sequence of steps, we crucially relate it to the quantity
\begin{equation*}
    \Tr((\mcO')^2 \E[\rho])
\end{equation*}
in order to establish exponential decay. Note that $\Tr((\mcO')^2 \E[\rho]) \leq \norm{\mcO'}^2_{\sf op} \cdot \Tr(\E[\rho]) \leq 1$. To our knowledge, all three of these components are new to our analysis. We leave it as an intriguing open question to find other applications of these ingredients to domains where ``biased Pauli analysis'' arises.

\begin{remark}[Predicting multiple observables]\label{remark:pre_shadows}
    {\em Just as in~\cite{huang2023learning}, we can also easily extend our guarantee to the setting where we wish to learn the joint mapping
    \begin{equation}
        (\rho, \mcO) \mapsto \Tr(\mcO\calE[\rho])\,. \label{eq:joint}
    \end{equation}
    This is the natural channel learning analogue of the question of classical shadows for state learning~\cite{huang2020predicting} \--- recall that in the latter setting, one would like to perform measurements on copies of $\rho$ and obliviously produce a classical description of the state that can then be used to compute some collection of observable values. We sketch the argument for extending to learning the joint mapping in Eq.~\eqref{eq:joint} in \Cref{remark:shadows}.}
\end{remark}

\paragraph{Impossibility for general concentrated distributions.} 

It is natural to wonder to what extent our results can be generalized, especially to states that are entangled. Could it be that all one needs is some kind of global covariance bound? Unfortunately, we show in \Cref{sec:lower-bounds} that even in the classical setting, this is not the case. Since classical distributions can be encoded by distributions over qubits,
this implies hardness for learning in the quantum setting as well.

\begin{theorem}[Hardness of learning over general concentrated distributions]
\label{thm:code-lb}
There exists a distribution $\mcD$ over $[-(1-\eta), 1-\eta]^n$ and a concept class $\mcC$ such that no algorithm PAC-learns the class $\mcC$ over $D$ in subexponential time.
\end{theorem}

\noindent There is a wide spectrum of distributional assumptions
that interpolates between fully product distributions and general concentrated distributions --- for instance, products of $k$-dimensional qudit distributions, output states of small quantum circuits, or distributions over negatively associated variables.
As discussed earlier, our understanding of when it is possible to predict the average-case behavior of arbitrary quantum dynamics is still nascent, and understanding learnability with respect to these more expressive distributional assumptions remains an important open question.

\paragraph{Organization.}
In \Cref{sec:prelims}, we state some preliminaries.
In \Cref{sec:classical}, we prove \Cref{thm:classical}, which is the classical setting and can be viewed as a warm-up to the quantum setting.
In \Cref{sec:quantum}, we prove \Cref{thm:quantum}, our main result.
Finally, in \Cref{sec:lower-bounds}, we show impossibility results including \Cref{thm:code-lb} and the failure of low-degree truncation in the standard Pauli basis.

\subsection{Related work}

Our work is part of a growing literature bridging classical computational learning theory and its quantum counterpart. Its motivation can be thought of as coming from the general area of quantum process tomography~\cite{mohseni2008quantum}, but as this is an incredibly extensive research direction, here we only focus our attention on surveying directly relevant works. 

\paragraph{Quantum analysis of Boolean functions.} 

In~\cite{montanaro2008quantum}, it was proposed to study Pauli decompositions of \emph{Hermitian} unitaries as the natural quantum analogue of Boolean functions. One notable follow-up work~\cite{rouze2022quantum} proved various quantum versions of classical Fourier analytic results like Talagrand's variance inequality and the KKL theorem in this setting, and also obtained corollaries about learning Hermitian unitaries in Frobenius norm given oracle access (see also~\cite{chen2023testing,bao2023nearly} for the non-Hermitian case). Recently, \cite{nadimpalli2023pauli} considered the Pauli spectrum of the \emph{Choi representation} of quantum channels and proved low-degree concentration for channels implemented by $\mathsf{QAC}^0$. One technical difference with our work is that these notions of Pauli decomposition are specific to the channel, whereas the object whose Pauli decomposition we consider is specific to the Heisenberg-evolved operator $\calE^\dagger[\mcO]$.
Additionally, we note that all of the above mentioned works focus on questions more akin to learning a full description of the channel and thus are inherently tied to channels with specific structure.

In contrast, our focus is on learning certain \emph{properties} of the channel, and only in an average-case sense over input states. As mentioned previously, this specific question was first studied in~\cite{huang2023learning}. There have been two direct follow-up works to this paper which are somewhat orthogonal to the thrust of our contributions. The first~\cite{volberg2023noncommutative} establishes refined versions of the so-called \emph{non-commutative Bohnenblust-Hille inequality} which was developed and leveraged by~\cite{huang2023learning} to obtain \emph{logarithmic} sample complexity bounds. In this work, we did not pursue this avenue of improvement but leave it as an interesting open question to improve our sample complexity guarantees accordingly. The second follow-up~\cite{klein2023quantum} to~\cite{huang2023learning} studies the natural qudit generalization of the original question where the distribution over qudits is similarly closed under a certain family of single-site transformations.


Finally, we note the recent work of~\cite{arunachalam2024learning} which studied the learnability of quantum channels with only low-degree Pauli coefficients. Their focus is incomparable to ours as they target a stronger metric for learning, namely $\ell_2$-distance for channels, but need to make a strong assumption on the complexity of the channel being learned. In contrast, we target a weaker metric, namely average-case error for predicting observables, but our guarantee applies for arbitrary channels.

\paragraph{Classical low-degree learning.}

The general technique of low-degree approximation in classical learning theory is too prevalent to do full justice to in this section. This idea of learning Boolean functions by approximating their low-degree Fourier truncation was first introduced in the seminal work of~\cite{linial1993constant}. Fourier-analytic techniques have been used to obtain new classical learning results for various concept classes like decision trees~\cite{kushilevitz1991learning}, linear threshold functions~\cite{kalai2008agnostically}, Boolean formulas~\cite{klivans2001learning}, low-degree polynomials \cite{eskenazis2022learning}, and more.

While Fourier analysis over biased distributions dates back to early work of Margulis and Russo~\cite{margulis1974probabilistic,russo1981critical,russo1982approximate}, it was first applied in a learning-theoretic context in~\cite{furst1991improved}, extending the aforementioned result of~\cite{linial1993constant}.

\section{Preliminaries}
\label{sec:prelims}

\subsection{Bloch sphere and Pauli covariance matrices}
\label{sec:bloch}

The Pauli matrices $I, X, Y, Z$ provide the basis for $2 \times 2$ Hermitian matrices.
This is captured by the following standard fact on expanding a single-qubit state using Pauli matrices.

\begin{fact}[Pauli expansion of states]
\label{fact:pauli expansion}
Any single-qubit mixed state $\rho$ can be written as
\[\rho = \frac{1}{2}(I + \alpha_xX + \alpha_yY + \alpha_zZ) \,,\]
where $\vec{\alpha} \in \R^3$, $\norm{\vec{\alpha}}_2 \le 1$, and $X,Y,Z$ are the standard Pauli matrices:
\[X = \begin{pmatrix}0&1\\
1&0\end{pmatrix}\quad Y = \begin{pmatrix}0&-i\\
i&0\end{pmatrix}\quad Z=\begin{pmatrix}1&0\\
0&-1\end{pmatrix}\,. \]
The set of all such $\vec{\alpha}$ of unit norm is the \emph{Bloch sphere}.

Any single-qubit distribution $D$ can be viewed as a distribution over the Bloch sphere.
We use $\E_D[\rho] \in \C^{2 \times 2}$ and $\vec{\mu} \in \R^3$ to refer to the expected state and the expected Bloch vector respectively.
\end{fact}

\noindent By taking tensor products of Pauli matrices, we obtain the collection of $4^n$ Pauli observables $\{I,X,Y,Z\}^{\otimes n}$, which form a basis for the space of $2^n \times 2^n$ Hermitian matrices:

\begin{fact}[Pauli expansion of observables]
Let $\mcO$ be an $n$-qubit observable. Then $\mcO$ can be written in the following form:
\[\mcO = \sum_{P \in \{I,X,Y,Z\}^{\otimes n}} \hat{\mcO}(P) \cdot P \,,\]
where $\hat{\mcO}(P) \triangleq \tr(\mcO P) / 2^n$. 
\end{fact}

\noindent Next, we define the Pauli covariance and second moment matrices associated to any distribution $D$ over the single-qubit Bloch sphere.

\begin{definition}[Pauli covariance and second moment matrices] \label{def:pauli-second-moment}
Let $D$ be a distribution over the Bloch sphere.
We will associate with $D$ the second moment matrix $\mcS \in \R^{3\times 3}$ and covariance matrix $\Sigma \in \R^{3\times 3}$,
indexed by the non-identity Pauli components $X, Y,$ and $Z$.
We define $\mcS$ such that $\mcS_{P,Q} = \E_{\rho \sim D}[\tr(P\rho)\tr(Q\rho)]$,
and $\Sigma$ such that $\Sigma_{P,Q} = \E_{\rho \sim D}[\tr(P\rho)\tr(Q\rho) - \tr(P\E[\rho])\tr(Q\E[\rho])]$.
\end{definition}

The following is a consequence of \Cref{fact:pauli expansion}:
\begin{fact}
For any distribution over the Bloch sphere, the Pauli second moment matrix $\mcS$ satisfies $\tr(\mcS) = 1$.
\end{fact}

\subsection{Access model}
\label{sec:access model}
In this paper we consider the following, standard access model, see e.g.~\cite{huang2023learning}. We assume that we can interact with the unknown channel $\calE$ by preparing any input state, passing it through $\calE$, and performing a measurement on the output. Additionally, we are given access to training examples from some distribution $\mcD = D^{\otimes n}$ over product states, and their corresponding \emph{classical descriptions}.
Because product states over $n$ qubits can be efficiently represent efficiently using $O(n)$ bits, the training set can be stored efficiently on classical computers.
The standard approach to represent a product state on a classical computer is as follows.
For each state $\rho = \otimes^n_{i=1} \ket{\psi_i}$ sampled from $\calD$, the classical description can be given by their 1-qubit Pauli expectation values: $\Tr(P \lvert \psi_i \rangle\!\langle \psi_i \rvert)$ for all $i\in[n]$ ranging over each qubit and $P\in\{X,Y,Z\}$.

Given these classical samples from $\mcD$ and the ability to query $\calE$, the learning goal is to produce a hypothesis $f'$ which takes as input the classical description of a product state $\rho$ and outputs an estimate for $\Tr(\mcO\calE[\rho])$. Formally, we want this hypothesis to have small test loss in the sense that $\E_{\rho\sim\mcD}[(\Tr(\mcO\calE[\rho]) - \Tr(f'(\rho)))^2] \le \epsilon$ with probability at least $1 - \delta$ over the randomness of the learning algorithm and the training examples from $\mcD$.

\subsection{Generalization bounds for learning}

For our learning protocol, we will use the following elementary results about linear and polynomial regression:

\begin{fact}[Rademacher complexity generalization bound \cite{mohri2018foundations}]
Let $\mcF$ be the class of bounded linear functions $[-1,1]^d\to[-B,B]$, and let $\ell$ be a loss function with Lipschitz constant $L$ and a uniform upper bound of $c$. 
With probability $1-\delta$ over the choice of a training set $S$ of size $m$ drawn i.i.d. from distribution $\mcD$,
\[\E_{(x,y) \sim \mcD}[\ell(f(x),y)] \le \E_{(x,y) \sim S}[\ell(f(x),y)] + 4LB\sqrt{\frac{d}{m}} + 2c \sqrt{\frac{\log 1/\delta}{2m}}\,.\]
\end{fact}


\begin{corollary}
\label{cor:sample complexity}
Let $f$ be a function that is $\eps$-close to a degree-$\le d$ polynomial $f^\star$:
\[\E_{x \sim \mcD}[(f(x)-f^\star(x))^2] \le \eps.\]
Then linear regression over the set of degree-$d$ polynomials with coefficients in $[-1,1]$ has time and sample complexity $\poly(n^d, \log 1/\eps) \cdot \log 1/\delta$ and finds $h$ such that 
\[\E_{x \sim \mcD}[(f(x)-h(x))^2] \le O(\eps)\]
with probability $1-\delta$.
\end{corollary}

\section{Warm-up: the classical case}
\label{sec:classical}
In this section we will prove \Cref{thm:classical}, which is a special case of \Cref{thm:quantum} where the distribution is classical, i.e. supported only on the $Z$ component. 

The key ingredient in the proof is to show that for any function which is $L^2$-integrable with respect to a product distribution over $[-(1-\eta),1-\eta]^n$ and whose extension to the hypercube is bounded, the function admits a ``low-degree'' approximation under an appropriate orthonormal basis. Roughly, the intuition is that the space of linear functions over a distribution $D$ on $[-(1-\eta), 1-\eta]$ has an orthonormal basis which is a $(1-\eta)$-scaling of a basis for a distribution on $\{-1,1\}$.
Therefore, the space of multilinear functions over the corresponding product distribution $D^{\otimes n}$ has a basis whose degree-$d$ components are scaled by $(1-\eta)^d$.
This, combined with the assumption that the function is bounded on the hypercube, allows us to conclude that the contribution of the degree-$d$ component to the variance of $f$ over $D^{\otimes n}$ is at most $(1-\eta)^{2d}$.

We prove this structural result in \Cref{sec:classical_lowdegree} and conclude the proof of \Cref{thm:classical} in \Cref{sec:classical_final}.

\subsection{Existence of low-degree approximation}
\label{sec:classical_lowdegree}

We first review some basic facts about classical biased Fourier analysis. For a more extensive overview of this topic, we refer the reader to \cite[Chapter 8]{ODonnell14}. 

Given a measure $\mu$, we let $L^2(\mu)$ denote the space of $L^2$-integrable functions with respect to $\mu$.

\begin{fact}[Biased Fourier basis]
\label{fact:biased basis}
Let $D$ be a distribution over $\bits$ with mean $\mu \in (-1,1)$. Given $f\in L_2(D^{\otimes n})$, the $\mu$-biased Fourier expansion of $f: \bits^n \to \R$ is 
\[f(x) = \sum_{S \subseteq [n]} \hat{f}(S) \phi_S(x) \,,\]
where $\phi_S(x) = \prod_{i \in S} \frac{x_i - \mu}{\sqrt{1-\mu^2}}$ and 
$\hat{f}(S) = \E_{x\sim D^{\otimes n}}[\phi_S(x)f(x)]$.
\end{fact}

\noindent The functions $\phi_S$ provide an orthonormal basis for the space of functions $L^2(D^{\otimes n})$, where $D$ is the distribution over $\{1, -1\}$ with mean $\mu$. We can naturally extend this to arbitrary product distributions over $\R^d$ as follows:


\begin{fact}[Basis for an arbitrary product distribution]
\label{fact:basis-classical}
Let $D$ be a distribution over $\R$ with mean $\mu$ and variance $\sigma^2 > 0$.
Then $\{1, \frac{x - \mu}{\sigma}\}$ is an orthonormal basis for $L^2(D)$,
and thus $\{1, \frac{x - \mu}{\sigma}\}^{\otimes n}$ is an orthonormal basis for $L^2(D^{\otimes n})$.
\end{fact}

The orthonormality of the basis immediately implies the following simple fact:

\begin{fact}[Parseval’s Theorem]
\label{fact:parseval}
For any function $f$ expressed as $f(x) = \sum_{S \subseteq [n]} \wh{f}(S) \phi_S(x)$, we have
$\E_{x\sim D^{\otimes n}}[f(x)^2] = \sum_{S\subseteq [n]} \wh{f}(S)^2$.
\end{fact}

\noindent The following is the crucial structural result in the classical setting that gives rise to \Cref{thm:classical}. Roughly speaking, it ensures that for the ``concentrated product distributions'' $D$ considered therein, any bounded multilinear function has decaying coefficients when expanded in the orthonormal basis for $L^2(D^{\otimes n})$.

\begin{lemma}
\label{lem:classical-approx}
Let $f:[-1,1]^n \to [-1,1]$ be a multilinear function and let $D$ be a distribution over $\R$ with mean $\mu$ and variance $\sigma^2 < 1-\mu^2$.
Then there exists a function $f^{\le d}$ such that 
\[\E_{x \sim D^{\otimes n}}[(f(x) - f^{\le d}(x))^2] \le \Paren{ \frac{\sigma^2}{1-\mu^2} }^d.\]
\end{lemma}

\begin{proof}
Let $f$ be expressed in the basis $B_{\mathrm{hypercube}} = \{1, \frac{x - \mu}{\sqrt{1-\mu^2}}\}^{\otimes n}$; i.e. as 
\[f(x) = \sum_{S \subseteq [n]} \wh{f}(S)\psi_S(x)\]
where $\psi_S = \prod_{i \in S} \frac{x_i - \mu}{\sqrt{1-\mu^2}}$.
Note that $\{\psi_S\}_{S\subseteq [n]}$ is orthonormal with respect to the distribution $\wt{D}^{\otimes n}$ where $\wt{D}$ is supported on $\{\pm1\}$ with mean $\mu$.
Since $|f(x)| \leq 1$ for $x \in [-1,1]^n$, it follows that
\begin{align*}
    1 \geq \E_{x\sim \wt{D}^{\otimes n}}[f(x)^2] = \sum_{S \subseteq [n]} \wh{f}(S)^2
\end{align*}
via \Cref{fact:parseval}.

Now, consider the basis $\phi_S \coloneqq \prod_{i\in S} \frac{x_i-\mu}{\sigma} = (\frac{\sqrt{1-\mu^2}}{\sigma})^{|S|} \cdot \psi_S$.
By \Cref{fact:basis-classical}, we know that $\{\phi_S\}_{S\subseteq[n]}$ is orthonormal with respect to $D$.
Let $f^{>d} \coloneqq \sum_{|S| > d} \wh{f}(S)\psi_S$.
We have
\begin{align*}
    \E_{x \sim D^{\otimes n}}[f^{>d}(x)^2] &= \E_{x \sim D^{\otimes n}}\Big[\Big(\sum_{|S| > d} \wh{f}(S)\psi_S\Big)^2\Big] \\ 
      &= \E_{x \sim D^{\otimes n}}\Big[\Big(\sum_{|S| > d} \wh{f}(S) \Big(\frac{\sigma}{\sqrt{1-\mu^2}}\Big)^{|S|} \phi_{S}\Big)^2 \Big] \\
      &= \sum_{|S|>d} \wh{f}(S)^2 \Big(\frac{\sigma^2}{1-\mu^2}\Big)^{|S|} \\
      &\le \Big(\frac{\sigma^2}{1-\mu^2}\Big)^{d} \sum_{|S|>d} \wh{f}(S)^2 \,,
\end{align*}
again using \Cref{fact:parseval}.
Since $\sum_{S} \wh{f}(S)^2 \leq 1$,
we have
\[\E_{x \sim D^{\otimes n}}[(f(x) - f^{\le d}(x))^2] \le \Big(\frac{\sigma^2}{1-\mu^2}\Big)^d\]
as desired.
\end{proof}

\subsection{Sample complexity and error analysis}
\label{sec:classical_final}

In light of \Cref{lem:classical-approx}, the proof of \Cref{thm:classical} is straightforward given the following elementary fact:

\begin{fact}[Bernoulli maximizes variance]
\label{fact:Bernoulli}
    Let $D$ be a distribution over the interval $[-(1-\eta), 1-\eta]$ with mean $\mu$.
    Then $\Var_{x \sim D}(x) \le (1-\eta)^2(1-\mu^2)$.
\end{fact}

We can now conclude the proof of \Cref{thm:classical}:

\begin{proof}[Proof of \Cref{thm:classical}]
By \Cref{fact:Bernoulli}, we have $\Var_{D}(x) \le (1-\eta)^2(1-\mu^2)$.
Then \Cref{lem:classical-approx} gives us a degree-$d$ approximation $f^{\le d}$ to $f$ such that
\[\E_{x \sim D^{\otimes n}}[(f(x)-f^{\le d}(x))^2] \le (1 -\eta)^{2d}.\]
Taking $d \coloneqq \log(1/\eps)/\log(1/(1-\eta))$
gives an approximation with error $\le \eps$.
Then by \Cref{cor:sample complexity}, linear regression on the space of polynomials of degree $\log(1/\eps)/\log(1/(1-\eta))$ finds an $O(\eps)$-error hypothesis in 
$n^{O(\log(1/\eps)/\log(1/(1-\eta)))} \cdot \log (1/\delta)$ time and samples.
\end{proof}

\begin{remark}
In the classical setting linear regression is not required, 
as we can estimate the mean, and the coefficients satisfy 
\[\hat{f}(S) = \E_{x \sim D^{\otimes n}} [f(x) \phi_S(x)] \,,\]
so they can be estimated directly.
We give the guarantee in terms of linear regression because the approach of directly estimating the coefficients does not generalize to the quantum setting.
\end{remark}


\section{Learning an unknown quantum channel}
\label{sec:quantum}

In this section we will prove \Cref{thm:quantum}.
First we will show that under any product distribution with second moment $\mcS$ such that $\|\mcS\|_{\mathsf{op}} \le 1-\eta$ for some $\eta \in (0,1)$,
every observable has a low-degree approximation.
A distribution has $\|\mcS\|_{\mathsf{op}} = 1$ only if it is effectively a classical distribution; i.e. it is supported on two antipodal points in the Bloch sphere.
So we are showing that any product distribution which is ``spread out'' within the Bloch sphere behaves well with low-degree approximation.

To do this, we cannot use exactly the same argument as in \Cref{sec:classical}, because 
there is not necessarily an orthonormal basis for our product distribution $D^{\otimes n}$ that is a ``stretched'' basis for some other distribution over the Bloch sphere.
Instead, we compare the variance of the observable under $D^{\otimes n}$ to the quantity $\E_{\rho\sim D^{\otimes n}}[\tr(\mcO^2 \rho)]$.
This allows us to use the boundedness of $\mcO$ to derive bounds for the contribution of the degree-$d$ part to the variance of $\mcO$ under $D^{\otimes n}$.
The learning algorithm will find the low-degree approximation by linear regression
over the degree-$\log(1/\eps)$ Pauli coefficients. The notion of low-degreeness will be with respect to a basis adapted to $D^{\otimes n}$.

\begin{definition}
\label{def:degree}
Let $D$ be a distribution over the Bloch sphere.
Let $U^\dagger A U$ be the eigendecomposition of $\ol{\rho} = \E_D[\rho]$. Let
    $\tX, \tY, \tZ := U^\dagger X U, U^\dagger Y U, \frac{U^\dagger Z U - \tr(\ol{\rho} U^\dagger Z U) I}{\sqrt{1 - \tr(\ol{\rho} U^\dagger Z U)^2}}$.
    
    Let $B = \{I, \tX, \tY, \tZ\}^{\otimes n}$.
    The degree of $P \in B$ is the number of non-identity elements in the product. The degree of a linear combination $\sum_{P\in B}\alpha_P P$ of elements in $B$ is the largest degree of $P \in B$ such that $\alpha_P \neq 0$.
\end{definition}

\noindent The fact that the degree is defined for any observable (which is equivalent to $\tilde X, \tilde Y, \tilde Z$ forming a basis of the space of operators) is the content of \Cref{clm:qubit basis}. The existence of low-degree approximation is guaranteed by the following lemma.

\begin{lemma}[Low-degree approximation]
    \label{lem:low-degree-approx}
Let $\mcO$ be a bounded $n$-qubit observable and 
let $D$ be a distribution over the Bloch sphere with mean $\vec{\mu}$ and Pauli second moment matrix $\mcS$ such that $\|\mcS\| \le 1-\eta$ for some $\eta \in (0,1)$.
Then there exists a degree-$d$ observable $\mcO^{\le d}$ and a constant $\eta' \in (0,1)$ such that 
\[\E_{\rho \sim D^{\otimes n}} [(\tr(\mcO\rho) - \tr(\mcO^{\le d}\rho))^2] \le (1-\eta')^d,\]
where $\eta'$ is a function of $\eta$.
\end{lemma}

\noindent Once \Cref{lem:low-degree-approx} is established, \Cref{thm:quantum} follows readily from an application of \Cref{cor:sample complexity}.

\begin{proof}[Proof of \Cref{thm:quantum}, using \Cref{lem:low-degree-approx}]
Let $1-\eta$ be a known upper bound on $\|\mcS\|_{\mathsf{op}}$.
We assume the access model of \Cref{sec:access model}, where we get a set $S$ of examples $[\tr(P\rho)]$ for 1-local $P$.
Our algorithm is as follows:
\begin{enumerate}
    \item Compute $\eta'$ as in the last line of \Cref{clm:op-norm-bound}. Let $d \coloneqq O(\log(1/\eps)/\log(1/(1-\eta')) $.
    \item Draw a set $S$ of size $n^{d} \cdot \log (1/\delta)$, and initialize $S'$ to be empty.
    \item For each $x \in S$, prepare a set $T$ of $\log(|S| + 1/\eps^2 + 1/\delta)$ copies of the state $\rho$ that matches the 1-local expectations of $x$. Let $\mathrm{est}(\tr(\mcO\mcE[\rho]))$ be the estimate of $\tr(\mcO\mcE[\rho])$, where for each $\rho \in T$, we measure with respect to $\{\mcE^\dagger[\mcO], I - \mcE^\dagger[\mcO]\}$, and $\mathrm{est}(\tr(\mcO\mcE[\rho]))$ is the empirical probability of measuring the first outcome.
    Add 
    \[(x^{\otimes \log(1/\eps)/\log(1/(1-\eta'))}, \mathrm{est}(\tr(\mcO\mcE[\rho])))\]
    to the set $S'$.
    \item Run linear regression on $S'$ and output the returned hypothesis $h$.
\end{enumerate}
By Hoeffding's inequality, each estimate of $\tr(\mcO\mcE[\rho])$ is within $\eps$ of its expectation with probability $1-\exp(-\Omega(|T| \eps^2))$.
By union bound over the $|S|$ estimates,
with probability $\ge 1 - \delta$, all estimates are within $\eps$ of $\tr(\mcO\mcE[\rho])$.

The time, sample, and error bounds follow from \Cref{lem:low-degree-approx}, which guarantees that $\mcE^\dagger[\mcO]$ is $\eps$-close to some degree-$d$ polynomial. This implies
the labels of our sample set are $2\eps$-close to such a polynomial.
The dimension of the linear regression problem is $\le \binom nd \cdot 3^d \le n^{O(d)}$, as there are 3 choices for each non-identity component.
Then by \Cref{cor:sample complexity}, linear regression has time and sample complexity 
\[n^{O(\log(1/\eps)/\log(1/(1-\eta')))} \cdot \log (1/\delta) \]
and outputs a hypothesis $h$ such that
\begin{align*}\E_{\rho \sim D^{\otimes n}} [(\tr(\mcO\mcE[\rho]) - \tr(h\rho))^2] &\le O(\eps)\,.\qedhere\end{align*}
\end{proof}

\noindent The proof of \Cref{lem:low-degree-approx} follows from three technical Lemmas. The first gives, for any product distribution over $n$-qubit states, a (non-orthogonal) decomposition of any observable into operators which are centered and bounded in variance with respect to that distribution.

\begin{lemma}
\label{clm:qubit basis}
Let $D$ be a distribution over the Bloch sphere.
Let $U^\dagger A U$ be the eigendecomposition of $\ol{\rho} = \E_D[\rho]$, and let 
$\tX, \tY, \tZ := U^\dagger X U, U^\dagger Y U, \frac{U^\dagger Z U - \tr(\ol{\rho} U^\dagger Z U) I}{\sqrt{1 - \tr(\ol{\rho} U^\dagger Z U)^2}}$.
Then $\{I, \tX, \tY, \tZ\}$ is a basis for the set of $2\times 2$ unitary Hermitian matrices, and thus every $n$-qubit observable $\mcO$ can be written as 
\[\mcO = \sum_{P \in B} \hat{\mcO}(P) P\]
for $B = \{I, \tX, \tY, \tZ\}^{\otimes n}$.
Furthermore, each non-identity $P \in B$ satisfies $\E_{\rho \sim D^{\otimes n}}[\tr(P\rho)] = 0$ and $\E_{\rho \sim D^{\otimes n}}[\tr(P\rho)^2] \le 1$.
\end{lemma}

\begin{proof}
It is clear that $\{I, \tX, \tY, \tZ\}$ is linearly independent and thus $B$ forms a basis for any $n$-qubit observable.
Note that $\Tr(P\rho) = \prod_{i=1}^n \Tr(P_i \rho_i)$ as $\rho$ is a product state, so we can restrict the analysis to single qubit states drawn from $D$.

For $P = \tX$ or $\tY$, it is clear that $\E_{\rho \sim D}[\tr(P\rho)] = 0$ and $\E_{\rho \sim D}[\tr(P\rho)^2] = 1$.
For $\tZ$, we have
\begin{align*}
\E_{\rho\sim D}[\tr(\tZ \rho)] &= \frac{\E[\tr (\rho U^\dagger Z U)] - \tr(\ol{\rho} U^\dagger Z U)}{\sqrt{1-\tr(\ol{\rho} U^\dagger Z U)^2}} = 0 \,.
\end{align*}
Moreover,
\begin{align*}
\E_{\rho\sim D}[\tr(\tilde{Z} \rho)^2] &= \frac{\Var[\tr(\rho U^\dagger Z U)]}{1-\tr(\ol{\rho} U^\dagger Z U)^2} \le 1 \,,
\end{align*}
because $\Var[\tr(\rho U^\dagger Z U)] + \E[\tr(\rho U^\dagger Z U)]^2 = \E [\tr(\rho U^\dagger Z U)^2] \le 1$.
\end{proof}
\begin{remark}
\label{remark:diagonalize}
Going forward, we will assume w.l.o.g. that the mean of the distribution $\E_D[\rho]$ is diagonal because we can always transform the basis according to $U$.
Thus, we will assume that the mean state $\E_D[\rho] = \frac{1}{2}(I+\mu Z)$ and the mean Bloch vector $\vec{\mu} = (0,0,\mu)$ for some $\mu \in [-1,1]$,
\end{remark}

\noindent We next prove two important components required in the proof of \Cref{lem:low-degree-approx}.
The first lemma gives an eigenvalue lower bound on a Hermitian matrix that arises when we expand $\Tr(\mcO^2 \E[\rho])$.

\begin{lemma} \label{clm:min-eigenvalue}
Let $\mu\in (-1,1)$ and
\[\wt{M} = \begin{pmatrix}
    1 & i\mu  \\
    -i\mu & 1 
\end{pmatrix}.\]
Then $\lambda_{\mathrm{min}}(\mathrm{Re}(\wt{M}^{\otimes k})) \ge (1-\mu^2)^{k/2}$ for any $k\in \N$.
\end{lemma}
\begin{proof}
Note that $\wt{M} = I + \mu Y$ and $Y$ is imaginary. Thus,
\begin{align*}
    \mathrm{Re}((I + \mu Y)^{\otimes k}) = \sum_{P\in \{I,Y\}^k: |P| \text{ even}} \mu^{|P|} \bigotimes_{i=1}^k P_i \,.
\end{align*}
Note also that the eigenvalues of the above remain the same if we replace the $Y$'s with $Z$'s.
Thus, the eigenvalues of the above can be indexed by $x\in \{\pm1\}^k$, and can be expressed as follows:
\begin{align*}
    \lambda(x) = \sum_{S\subseteq [k]: |S| \text{ even}} \mu^{|S|} x^S \,.
\end{align*}

Let $f: \R^k \to \R$ be the function $f(x) = \prod_{i=1}^k (1+x_i) = \sum_{S\subseteq[k]} x^S$. Then, we have $\lambda(x) = \frac{1}{2}(f(\mu x) + f(-\mu x))$.
By the AM-GM inequality, 
\begin{align*}\lambda(x) &\geq \sqrt{f(\mu x) f(-\mu x)} = \prod_{i=1}^k \sqrt{(1+\mu x_i) (1-\mu x_i)} = (1-\mu^2)^{k/2}\,.\qedhere\end{align*}
\end{proof}

%
%
%

\newcommand{\diagmat}{\Delta}

The next lemma gives an upper bound on the Pauli covariance matrix $\Sigma$ (recall \Cref{def:pauli-second-moment}) scaled by a specific diagonal matrix.

\begin{lemma}
\label{clm:op-norm-bound}
Let $D$ be a distribution over the Bloch sphere with mean $\vec{\mu} = (0,0,\mu)$, Pauli second moment matrix $\mcS$ such that $\|\mcS\|_{\mathsf{op}} \le 1-\eta$ for some $\eta \in (0,1)$, and Pauli covariance matrix $\Sigma$.
Let $\diagmat = \diag((1-\mu^2)^{-1/4},(1-\mu^2)^{-1/4},(1-\mu^2)^{-1/2})$. 
Then there exists $\eta' \in (0,1)$ such that $\|\diagmat\Sigma \diagmat\|_{\mathsf{op}} \le 1-\eta'$, where $\eta'$ is a function of $\eta$.

\end{lemma}
\begin{proof}
We split into two cases based on whether $\mu^2 \ge \eta/2$.
The case where $\mu^2 < \eta/2$
is the simpler case: since $\|\Sigma_{\sf op}\| \leq \|\mcS\|_{\sf op} \leq 1-\eta$ and $\|\diagmat\|_{\sf op}^2 \leq (1-\mu^2)^{-1}$, we have 
\begin{align*}
    \|\diagmat \Sigma \diagmat\|_{\sf op} \le \frac{\|\Sigma\|_{\sf op}}{1 - \mu^2} \le   \frac{1 - \eta}{1 - \mu^2}
    \le \frac{1-\eta}{1-\eta/2}
    \le 1-\eta/2 \,.
\end{align*}

Now we consider the case where the inequality is not satisfied.
We note that $\Sigma = \mcS - \vec{\mu}\vec{\mu}^\top$ and $\vec{\mu} = (0,0,\mu)$, thus we have that
$\diagmat\Sigma \diagmat$ has the following block structure:
\[\diagmat \Sigma \diagmat = \begin{pmatrix}
    \frac{\mcS_{2 \times 2}}{\sqrt{1 - \mu^2}} & b \\
    b^\dagger &\frac{\mcS_{zz} - \mu^2}{1 - \mu^2}
\end{pmatrix} \,, \]
where $\mcS_{2\times 2}$ is the top left $2 \times 2$ block of $\mcS$, $\mcS_{zz}$ is the bottom right entry, and $b$ is the remaining part (its values will not be directly relevant).
Note also that $\diagmat \Sigma \diagmat \succeq 0$, and it is well-known that any PSD matrix of the form $\begin{bmatrix}
    A & b \\ b^\dagger & c
\end{bmatrix} \succeq 0$ has operator norm at most $\norm{A}_{\sf op} + c$.
We thus know that 
\[\|\diagmat \Sigma \diagmat\|_{\sf op} \le \frac{\mcS_{zz} - \mu^2}{1 - \mu^2} + \frac{\|\mcS_{2 \times 2}\|_{\sf op}}{\sqrt{1 - \mu^2}} 
\leq \frac{\mcS_{zz} - \mu^2}{1 - \mu^2} + \frac{\tr(\mcS_{2\times 2})}{\sqrt{1 - \mu^2}} 
\le  \frac{\mcS_{zz} - \mu^2}{1 - \mu^2} + \frac{1 - \mcS_{zz}}{\sqrt{1 - \mu^2}} 
\,,\]
where we use the fact that $1 = \tr(\mcS) = \tr(\mcS_{2\times 2}) + \mcS_{zz}$.
Then we have 
\begin{align*} 
\|\diagmat \Sigma \diagmat\|_{\sf op}
&\le  \frac{\mcS_{zz} - \mu^2}{1 - \mu^2} + \frac{1 - \mcS_{zz}}{\sqrt{1 - \mu^2}}  \\
&= \mcS_{zz}\Big(\frac{1}{1-\mu^2} - \frac{1}{\sqrt{1-\mu^2}}\Big) + \frac{1}{\sqrt{1 - \mu^2}} - \frac{\mu^2}{1 - \mu^2}  \\
&\le (1-\eta)\Big(\frac{1}{1-\mu^2} - \frac{1}{\sqrt{1-\mu^2}}\Big) + \frac{1}{\sqrt{1 - \mu^2}} - \frac{\mu^2}{1 - \mu^2}  \\
&= 1 - \eta\Big(\frac{1}{1-\mu^2}-\frac{1}{\sqrt{1-\mu^2}}\Big) \\
&\le 1 - \eta \frac{1 -\sqrt{1-\eta/2}}{1-\eta/2} \,.
\end{align*}
The third line is by the fact that $\mcS_{zz} \le \|\mcS\|_{\mathsf{op}} \le 1 - \eta$,
and the last inequality is because the function $\frac{1}{1-\mu^2} - \frac{1}{\sqrt{1-\mu^2}}$ is increasing with $\mu^2$ and that we are in the $\mu^2 \geq \eta/2$ case.

Thus, combining the two cases, there always exists $\eta'$ such that
$\|\diagmat \Sigma \diagmat\|_{\sf op} \le 1 - \eta'$.
Specifically, we have 
\begin{align*}\eta' &\ge \min\Big\{\eta \frac{1 -\sqrt{1-\eta/2}}{1-\eta/2},\ \eta/2 \Big\} > 0 \,.
\qedhere\end{align*}
\end{proof}

\begin{remark}\label{remark:shadows}
    As mentioned in \Cref{remark:pre_shadows}, our techniques can be extended to learn not just the mapping $\rho \mapsto \Tr(O\mcE[\rho])$, but also the joint mapping $(O,\rho)\mapsto \Tr(O\mcE[\rho])$. As this is standard, here we only briefly sketch the main ideas.
    
    The general strategy is to produce a classical description of the channel that we can then use to make predictions about properties of output states. To do this, we draw many input states $\rho_1,\ldots,\rho_N$ from $\calD$, query the channel on each them, and for each of the output states $\calE[\rho_j]$, we apply a randomized Pauli measurement to each of them and use these to form unbiased estimators for the output state. Concretely, given an output state $\calE[\rho_j]$, a randomized Pauli measurement will result in a stabilizer state $\ket{\psi^{(j)}} = \otimes^n_{i=1} \ket{s^{(j)}_i} \in \{\ket{0},\ket{1},\ket{+},\ket{-},\ket{y+},\ket{y-}\}^{\otimes n}$, and the expectation of $\otimes^n_{i=1} (3\ket{s^{(j)}_i}\!\bra{s^{(j)}_i} - I)$ is $\calE[\rho_j]$. The classical description of the channel is given by the $O(\log(1/\epsilon))$-body reduced density matrices of the input states $\rho_1,\ldots,\rho_N$, together with the classical encodings of $\ket{\psi^{(1)}},\ldots,\ket{\psi^{(N)}}$.

    Given an observable $O$, we can then perform regression to predict the labels $\Tr(O\otimes^n_{i=1}(3\ket{s^{(j)}_i}\!\bra{s^{(j)}_i}-I))$ given the features $\{\Tr(P\rho_j)\}_{|P| \le O(\log (1/\epsilon))}$. Because the labels are unbiased estimates of $\Tr(O\calE[\rho_j])$, the resulting estimator will be an accurate approximation to $\rho \mapsto \Tr(O\calE[\rho])$ for $\rho\sim\calD$. 
    
    In this work, we do not belabor these details as they are already investigated in depth in~\cite{huang2023learning}. Instead, we focus on the single observable case as this is where the main difficulty lies in extending the results of~\cite{huang2023learning} to more general input distributions.
\end{remark}

\medskip
\noindent Now we prove \Cref{lem:low-degree-approx} which shows the existence of a low-degree approximator.

\begin{proof}[Proof of \Cref{lem:low-degree-approx}]
Assume w.l.o.g., as in \Cref{remark:diagonalize}, that $\E_D[\rho] = \frac{1}{2}(I + \mu Z)$ and $\vec{\mu} = (0,0,\mu)$.
Let $\mcO$ be expressed in the basis $B = \{I, X, Y, \frac{Z - \mu I}{\sqrt{1-\mu^2}}\}^{\otimes n}$ as in \Cref{clm:qubit basis}.
For a subset $S \subseteq [n]$, we will denote by $\{P \in B: P \sim S\}$ the set of basis elements with non-identity components at indices in $S$ and
identity components at indices in $\overline{S}$.
We will also denote by $|P|$ the number of non-identity components in $P$.

Let $\mcO^{> d} \coloneqq \sum_{|S| > d}\sum_{P \sim S} \hat{\mcO}(S)P$.
We will show that $\mcO^{>d}$ satisfies $\E_{\rho \sim D^{\otimes n}}[\tr(\mcO^{> d} \rho)^2] \le (1 - \eta')^d$ where $\eta' \in (0,1)$ depends on $\eta$.

We first note a bound on the related quantity $\E_{\rho\sim D^{\otimes n}} [\tr(\mcO^2 \rho)]$. 
Because $\|\mcO\|_{\sf op} \le 1$ and $\tr(\E[\rho]) \le 1$, $\tr(\mcO^2 \E[\rho]) \le 1$ as well.
We will expand this quantity as a quadratic form.
\begin{align*}
    \E_{\rho\sim D^{\otimes n}} [\tr(\mcO^2 \rho)] &= \sum_{P,Q \in B} \hat{\mcO}(P) \hat{\mcO}(Q)\tr(PQ \E[\rho]) \\
    &=  \sum_{P,Q \in B} \hat{\mcO}(P) \hat{\mcO}(Q)\tr(\otimes_{i=1}^n P_iQ_i \E[\rho_i]) \\
    &=  \sum_{P,Q \in B} \hat{\mcO}(P) \hat{\mcO}(Q)\prod_{i=1}^n\tr( P_iQ_i \E[\rho_i]) \,.
\end{align*}
Note that the product is 0 whenever exactly one of $P_i$, $Q_i$ is $I$ for any $i \in [n]$ (as $\Tr((Z-\mu I)\cdot (I+\mu Z)) = 0$).
Therefore, we can partition the terms into groups that share a subset of identity variables: 
\begin{align*}
    \E_{\rho\sim D^{\otimes n}} [\tr(\mcO^2 \rho)] &= \sum_{S \subseteq [n]} \sum_{P,Q \sim S} \hat{\mcO}(P) \hat{\mcO}(Q)\prod_{i=1}^n\tr( P_iQ_i \E[\rho_i]) \\
    &= \sum_{S \subseteq [n]} \hat{\mcO}_S^\dagger M^{\otimes |S|} \hat{\mcO}_S \,,
\end{align*}
where $\hat{\mcO}_S$ is the vector of coefficients for the set $\{P: P \sim S\}$, and $M$ is the $3 \times 3$ matrix such that $M_{P,Q} = \tr( PQ \E_D[\rho])$:
\[M = \begin{pmatrix}
    1 & i\mu  & 0  \\
    -i\mu  & 1 & 0  \\
    0 & 0 & 1 
\end{pmatrix} \,. \]
Here, the entry $i\mu$ arises because $XY = iZ$ and $\Tr(XY \cdot \frac{1}{2}(I+\mu Z)) = i\mu$.
Since $M$ is positive semidefinite,
we have $\hat{\mcO}_S^\dagger M^{\otimes |S|} \hat{\mcO}_S \ge 0$ for all $S$,
and therefore we have 
\begin{align}
    1 \ge \E_{\rho\sim D^{\otimes n}} [\tr(\mcO^2 \rho)]
    = \sum_{S \subseteq [n]} \hat{\mcO}_S^\dagger M^{\otimes |S|} \hat{\mcO}_S
    \ge \sum_{S \subseteq [n]:|S| > d} \hat{\mcO}_S^\dagger M^{\otimes |S|} \hat{\mcO}_S \,.
    \label{eq:tr-O2-rho}
\end{align}
Now we will write the desired quantity $\E_{\rho \sim D^{\otimes n}} [\tr(\mcO^{> d}\rho)^2]$ as a quadratic form as well:
\begin{align*}
    \E_{\rho \sim D^{\otimes n}} [\tr(\mcO^{> d} \rho)^2] &= \sum_{P,Q \subseteq B:\ |P|,|Q|>d} \hat{\mcO}(P) \hat{\mcO}(Q) \prod_{i \in [n]} \E[\tr(P_i \rho) \tr(Q_i \rho)] \,. 
\end{align*}

Similarly, the product is 0 whenever exactly one of $P_i$ and $Q_i$ is identity (by the guarantee of \Cref{clm:qubit basis} that $\E[\tr(P_i\rho)] = 0$ for $P_i \neq I$), so we can make the same partition:

\begin{align}
    \E_{\rho \sim D^{\otimes n}} [\tr(\mcO^{>d} \rho)^2] &= \sum_{|S|>d} \hat{\mcO}_S^\dagger M'^{\otimes |S|} \hat{\mcO} \,.
    \label{eq:tr-O-rho-squared}
\end{align}

Here $M'$ is $3 \times 3$ matrix such that $M'_{PQ} = \E_{\rho\sim D}[\tr(P\rho)\tr(Q\rho)]$ for $P,Q \in \{X,Y,\frac{Z-\mu I}{\sqrt{1-\mu^2}}\}$; in other words, it is the second moment matrix of the non-identity elements of our biased Pauli basis.
Below, we show the entries of $M'$ in terms of the Pauli covariance matrix $\Sigma$:

\[M' = \begin{pmatrix}
    \Sigma_{xx} & \Sigma_{xy} & \frac{\Sigma_{xz}}{\sqrt{1 - \mu^2}} \\
    \Sigma_{xy} & \Sigma_{yy} & \frac{\Sigma_{yz}}{\sqrt{1 - \mu^2}} \\
    \frac{\Sigma_{xz}}{\sqrt{1 - \mu^2}} & \frac{\Sigma_{yz}}{\sqrt{1 - \mu^2}} & \frac{\Sigma_{zz}}{1 - \mu^2}
\end{pmatrix} \,. \]
Our aim is to bound $M'$ in terms of $M$ in order to show the existence of some $\eta' \in (0,1)$ such that  
\[\E[\tr(\mcO^{> d} \rho)^2] \le (1 - \eta')^d \cdot \tr((\mcO^{> d})^2 \E[\rho]) \le (1-\eta')^d.\]
Comparing Eq.~\eqref{eq:tr-O2-rho} and \eqref{eq:tr-O-rho-squared},
it suffices to show that 
\[(1 - \eta')^{|S|}\hat{\mcO}_S^\dagger M^{\otimes |S|} \hat{\mcO}_S \ge  \hat{\mcO}_S^\dagger M'^{\otimes |S|} \hat{\mcO}_S\]
for all vectors $\mcO_S$.
Crucially, since the Pauli coefficients $\mcO_S$ must be \emph{real}, it suffices to prove that
\begin{align}
    M'^{\otimes |S|} \preceq
    (1-\eta')^{|S|} \cdot \mathrm{Re}(M^{\otimes |S|}) \,.
    \label{eq:M'-tensor-bound}
\end{align}


Let $M_{2\times 2}$ be the top left $2 \times 2$ block in $M$, which is exactly the matrix in \Cref{clm:min-eigenvalue}.
From \Cref{clm:min-eigenvalue}, we have $\lambda_{\min}(\mathrm{Re}(M_{2\times 2}^{\otimes |S|})) \ge (1- \mu^2)^{ |S|/2 }$.
By the block structure of $M$, it follows that $\mathrm{Re}(M^{\otimes |S|}) \succeq \diag(\sqrt{1-\mu^2}, \sqrt{1-\mu^2}, 1)^{\otimes |S|}$.
Thus, we can establish Eq.~\eqref{eq:M'-tensor-bound} by proving that
\begin{align}
    M' \preceq (1-\eta') \cdot \diag\Big( \sqrt{1-\mu^2}, \sqrt{1-\mu^2}, 1 \Big) \,.
    \label{eq:M'-bound}
\end{align}

Now, note that $M'$ can be written as $\tilde{\diagmat} \Sigma \tilde{\diagmat}$, where $\Sigma$ is the covariance matrix of $D$ and $\tilde{\diagmat} = \diag(1,1, (1-\mu^2)^{-1/2})$.
Letting $\diagmat = \diag((1 - \mu^2)^{-1/4}, (1 - \mu^2)^{-1/4}, (1 - \mu^2)^{-1/2})$ (which is the diagonal matrix defined in \Cref{clm:op-norm-bound}), we can see that Eq.~\eqref{eq:M'-bound} is equivalent to
\[\|\diagmat M'\diagmat \|_{\mathsf{op}} \le 1-\eta',\]
By \Cref{clm:op-norm-bound}, this is true by our assumption that $\|\mcS\|_{\mathsf{op}} \le 1-\eta$.
This establishes Eq.~\eqref{eq:M'-bound} and thus Eq.~\eqref{eq:M'-tensor-bound}, and from Eq.~\eqref{eq:tr-O2-rho} and \eqref{eq:tr-O-rho-squared}, we have
\begin{align*}
    1 &\ge \E_{\rho\sim D^{\otimes n}} [\tr(\mcO^2 \rho)]
    \ge \sum_{|S| > d} \hat{\mcO}_S^\dagger M^{\otimes |S|} \hat{\mcO}_S \\
    &\ge (1-\eta')^{-d} \sum_{|S| > d} \hat{\mcO}_S^\dagger M'^{\otimes |S|} \hat{\mcO}_S \\
    &= (1-\eta')^{-d} \E_{\rho \sim D^{\otimes n}}[\tr(\mcO^{>d}\rho)^2] \,.
\end{align*}
Therefore,
\[\E_{\rho\sim D^{\otimes n}}[(\tr(\mcO\rho) - \tr(\mcO^{\le d}\rho))^2] \le (1-\eta')^d \,.\]
Hence, we have obtained the claimed result.
\end{proof}

\begin{remark} \label{rem:different-distributions}
    In the proof of \Cref{lem:low-degree-approx}, we relate the quantity $\E_{\rho \sim D^{\otimes n}}[\tr(\mcO^{>d}\rho)^2]$ that we wish to bound to the related quantity $\E_{\rho\sim D^{\otimes n}} [\tr(\mcO^2 \rho)]$, which is at most $1$ since $\|\mcO\|_{\sf op} \leq 1$.
    Due to the choice of our biased Pauli basis $B = \{I,X,Y, \frac{Z-\mu I}{\sqrt{1-\mu^2}}\}^{\otimes n}$, we may write both quantities as sums of $\hat{\mcO}_S^\dagger M'^{\otimes |S|} \hat{\mcO}_S$ and $\hat{\mcO}_S^\dagger M^{\otimes |S|} \hat{\mcO}_S$
    (see Eq.~\eqref{eq:tr-O2-rho} and \eqref{eq:tr-O-rho-squared}).

    Suppose $\rho$ is a product of different distributions where each qubit has mean Bloch vector $\vec{\mu}_i = (0,0,\mu_i)$ (after rotation; see \Cref{remark:diagonalize}) and second moment bounded by $1-\eta$.
    We will instead use the basis $B = \otimes_{i=1}^n \{I,X,Y, \frac{Z-\mu_i I}{\sqrt{1-\mu_i^2}}\}$.
    One can easily see that the above quantities are essentially the same, with $M^{\otimes |S|}$ replaced by $\otimes_{i\in S} M_i$ and similarly for $M'$.
    Then, the next steps in the proof (establishing Eq.~\eqref{eq:M'-tensor-bound} and \eqref{eq:M'-bound}) are exactly the same.
\end{remark}





\section{Lower bounds}
\label{sec:lower-bounds}

In this section, we show lower bounds in the classical case, which automatically imply hardness in the quantum case.
In \Cref{sec:lb-non-product}, we prove \Cref{thm:code-lb}, which shows hardness of learning without the product distribution assumption in \Cref{thm:classical}.
In \Cref{sec:lowerbound_unbiased_degree}, we show that truncating in the unbiased basis fails when the distribution is not mean zero.

\subsection{Lower bounds for learning non-product distributions}
\label{sec:lb-non-product}

In this section, we prove \Cref{thm:code-lb}. We show that if $\calD$ is an arbitrary distribution, then even if $\calD$ is supported on $[-(1-\eta), (1-\eta)]^n$ (i.e., in the interior of the hypercube), there is no learning algorithm.

Let $C$ be a code over $\zo^n$ of size $2^{\Theta(n)}$ and distance $n/4$.
The following fact is standard.

\begin{fact} \label{fact:hardness-learning-code}
    For any constant $\eps > 0$, any learning algorithm that can learn an arbitrary function $f: \{\pm1\}^n \to \{\pm1\}$ over $C$ to $\eps$ error requires $2^{\Omega(n)}$ queries.
\end{fact}


\begin{proof}[Proof of \Cref{thm:code-lb}]
    Let $\eta = 0.1$.
    We set the distribution $\calD$ to be the uniform distribution over $(1-\eta)\cdot C$, which is supported in $[-(1-\eta), (1-\eta)]^n$.
    Let $f: \{\pm1\}^n \to \{\pm1\}$ be an arbitrary function.
    Since the code $C$ has distance $n/4$, we can without loss of generality assume that $f(x') = f(x)$ whenever $d(x,x') \leq n/8$.

    Suppose for contradiction that there is an algorithm that, with only $2^{o(n)}$ queries to $f$, outputs a function $g: [-1,1]^n \to \R$ such that $\E_{x \sim \calD}[(g(x) - f(x))^2] \leq \eps$.
    Let $p \coloneqq 1-\eta/2$, and let $\mathrm{Ber}(p)$ be the distribution where we have $+1$ with probability $p$ and $-1$ otherwise.
    Then, for any $x\in \{\pm1\}^n$, we have $f((1-\eta)x) = \E_{z\sim \mathrm{Ber}(p)^{\otimes n}}[f(x \circ z)]$, where $\circ$ denotes entry-wise product.

    We claim that for any $x\in C$, $|\E_{z\sim \mathrm{Ber}(p)^{\otimes n}}[f(x \circ z)] - f(x)| \leq o_n(1)$.
    The number of $-1$ coordinates in $z$ is distributed as a $\mathrm{Bin}(n,\eta/2)$.
    By the Chernoff bound, $\Pr[\mathrm{Bin}(n,\eta/2) \geq (1+\delta)\frac{1}{2}\eta n] \leq e^{-O(\delta^2 \eta n)} = o_n(1)$.
    In particular, for $\eta = 0.1$, we have that $\Pr[\mathrm{Bin}(n,\eta/2) \geq n/8] \leq o_n(1)$.
    Thus, with probability $1-o_n(1)$, $d(x\circ z, x) < n/8$, which means that $f(x\circ z) = f(x)$.
    This proves that $|f((1-\eta)x) -f(x)| \leq o_n(1)$ for all $x\in C$.

    Then, let $h(x) = g((1-\eta)x)$.
    \begin{align*}
        \E_{x\in C}\Brac{(h(x) - f(x))^2}
        &= \E_{x\in C}\Brac{(g((1-\eta)x) - f(x))^2} \\
        &\leq \E_{x\in C}\Brac{(g((1-\eta)x) - f((1-\eta)x))^2} + o_n(1) \\
        &= \E_{x\sim \calD}\Brac{(g(x) - f(x))^2} + o_n(1)  \\
        &\leq \eps + o_n(1) \,.
    \end{align*}
    This means that $h$ is an $\eps$-approximation of $f$ over $C$.
    This contradicts \Cref{fact:hardness-learning-code} thus completing the proof.
\end{proof}

\subsection{Lower bounds for unbiased degree truncation}
\label{sec:lowerbound_unbiased_degree}

In \Cref{thm:classical}, if the product distribution $\mcD$ over $[-(1-\eta), 1-\eta]^n$ has mean zero, then directly truncating $f$ with respect to the standard monomial basis at degree $O(\log(1/\eps)/\log(1/(1-\eta)))$ (independent of $n$) suffices.
However, in this section, we will show that without the mean zero assumption, even truncating at $\Omega(n)$ degree w.r.t.\ the monomial basis does not give a small approximation error.
Our counter-example implies that truncation in any distribution-oblivious basis will fail on some product distribution.
In the quantum setting, this implies that low-degree truncation in the standard Pauli basis fails on some product distribution as well.

The counter-example is quite simple: $f$ is the multilinear extension of the Boolean majority function, and $\mcD$ is supported on a single nonzero point (which is in fact a product distribution).

\begin{fact}[\cite{ODonnell14}]
\label{fact:majority-weight}
Let $f$ be the multilinear extension of the Boolean majority function. 
Then the $\ell_2$ Fourier weight on terms of degree $k$ is $\Theta(k^{-3/2})$.
\end{fact}

The following is a well-known fact in approximation theory.

\begin{fact}[Chebyshev extremal polynomial inequality~\cite{rivlin2020chebyshev}]
\label{fact:leading-coeff}
    Let $p$ be a degree-$d$ univariate polynomial with leading coefficient $1$.
    Then, $\max_{x\in[-1,1]} |p(x)| \geq 2^{-d+1}$.
\end{fact}

\begin{lemma}
    Let $f: \{\pm1\}^n \to \{\pm1\}$ be the majority function extended to the domain $[-1,1]^n$.
    Let $0 < a < b < 1$ be fixed constants.
    Then, there exist $\delta \coloneqq \delta(a,b) \in (0,1)$ and $t^*\in [a,b]$ such that the degree-$\delta n$ truncation $f^{\leq \delta n}$ has $|f^{\leq \delta n}(t^* \cdot \vec{1})| \geq \omega_n(1)$.
\end{lemma}
\begin{proof}
    
    Let $g(t) \coloneqq f^{\leq d}(t\cdot \vec{1})$, which is a univariate polynomial of degree $d$, and consider the shifted polynomial $\wt{g}(t) = g(\frac{b-a}{2}t + \frac{a+b}{2})$.
    The leading coefficient of $\wt{g}$, denoted $\wt{c}_d$, is $c_d (\frac{b-a}{2})^d$, where $c_d$ is the leading coefficient of $g$.
    Since $g(t) = f^{\leq d}(t\cdot \vec{1})$, we have $c_d = \sum_{S: |S|=d} \wh{f}(S)$.
    For the majority function, \Cref{fact:majority-weight} implies that $\binom{n}{d} \wh{f}(S)^2 = \Theta(d^{-3/2})$ for all $S$ of size $d$, thus
    \begin{align*}
        |\wt{c}_d| = \Paren{\frac{b-a}{2}}^d \binom{n}{d}^{1/2} \Theta(d^{-3/4}) \,.
    \end{align*}
    Then, by \Cref{fact:leading-coeff}, there must be a $s^*\in[-1,1]$ such that
    \begin{align*}
        |\wt{g}(s^*)| \geq 2^{-d+1} \cdot  \Paren{\frac{b-a}{2}}^d \binom{n}{d}^{1/2} \Theta(d^{-3/4}) \,.
    \end{align*}
    If $d = \delta n$, then $\binom{n}{d} \geq (\frac{1}{\delta})^{d} = e^{d\log(1/\delta)}$.
    Thus, given $0 < a < b < 1$, there exists a $\delta \in (0,1)$ such that the above is $\exp({\Omega(d)})$.
    Thus, there exists a $t^*\in [a,b]$ such that $|f^{\leq d}(t^* \cdot \vec{1})| \geq \omega_n(1)$.
\end{proof}

\section*{Acknowledgments}

SC and HH would like to thank Ryan O'Donnell for a helpful discussion at an early stage of this project. Much of this work was completed while JD and JL were interns at Microsoft Research.

\bibliographystyle{alpha}
\bibliography{biblio}

\end{document}